\documentclass{sigplanconf}

\usepackage[utf8]{inputenc}

\usepackage[english]{babel} 
\usepackage{amsmath} 
\usepackage{amsfonts} 
\usepackage{amssymb} 
\usepackage{amsthm}
\usepackage{mathdots}    
\usepackage{url} 
\usepackage{semantic}
\usepackage{cmll}
\usepackage{subfigure}
\usepackage{graphics,color}      
\usepackage[shortlabels]{enumitem}

\usepackage{bussproofs}
\usepackage{tikz}
\usetikzlibrary{arrows,chains,matrix,positioning,scopes}
\usepackage[final, backref, hyperindex, bookmarks]{hyperref}

\newcommand{\sub}[2]{#1[#2]}
\newcommand{\limp}{\multimap} 
\newcommand{\bang}{\oc} 
\newcommand{\limpe}[2]{\multimap}
\newcommand{\hyp}[4]{#1:(#2, #3)}
\newcommand{\letm}[3]{{\sf let} \ ! #1 = #2 \ {\sf in} \ #3}    
\newcommand{\tertype}{{\bf 1}}
\newcommand{\behtype}{{\bf B}}

\newcommand{\set}[1]{\{#1\}}
\newcommand{\st}[2]{{\sf set}(#1,#2)}

\newcommand{\rgtype}[2]{{\it {\sf Reg}_{#1} #2}}
\newcommand{\get}[1]{{\sf get}(#1)}
\newcommand{\sff}{{\sf ob}}
\newcommand{\sv}{{\sf v}}
\newcommand{\getsf}{{\sf get}}
\newcommand{\getsfv}{{\sf get_{\sv}}}

\newcommand{\setsf}{{\sf set}}
\newcommand{\setsfv}{{\sf set_{\sv}}}

\newcommand{\gcsf}{{\sf gc}}
\newcommand{\gcsfv}{{\sf gc_{\sv}}}

\newcommand{\store}[2]{#1 \Leftarrow #2}
\newcommand{\regtype}[2]{{\sf Reg}_{#1} #2}

\newcommand{\letb}[3]{\mathsf{let}\;\oc #1 = #2\;\mathsf{in}\;#3}
\newcommand{\letp}[3]{\mathsf{let}\;\mathsection #1 = #2\;\mathsf{in}\;#3}
\newcommand{\letd}[3]{\mathsf{let}\;\dagger\negmedspace #1 = #2\;\mathsf{in}\;#3}

\newcommand{\nat}{\mathsf{Nat}}
\newcommand{\bnat}{\mathsf{BNat}}

\newcommand{\reduc}{\longrightarrow}
\newcommand{\reducf}{\longrightarrow_{\sff}}
\newcommand{\treduc}{\longrightarrow^{*}}
\newcommand{\treducf}{\longrightarrow^{*}_{\sff}}
\newcommand{\preduc}{\longrightarrow^{+}}

 \newcommand{\reduci}[1]{\stackrel{#1}{\reduc}}
 \newcommand{\reducif}[1]{\stackrel{#1}{\reducf}}

 \newcommand{\treduci}[1]{\stackrel{#1}{\treduc}}
 
\newcommand{\para}{\parallel}
\newcommand{\churchn}[1]{\overline{#1}}
\newcommand{\letbang}{\mathsf{let}\,\oc}

\newcommand{\listtype}[1]{\mathsf{List}\:#1}

\newcommand{\lit}{\mathsf{list\_it}}
\newcommand{\churchl}[1]{[#1]}
\newcommand{\add}{\mathsf{add}}

\newcommand{\fv}[1]{\mathsf{FV}(#1)}

\newcommand{\fo}[2]{\mathsf{FO}(#1,#2)}
\newcommand{\foa}[1]{\mathsf{FO}(#1)}
\newcommand{\ie}{\emph{i.e.}\;}

\newcommand{\msec}{\mathsection}

\newcommand{\uaff}{\lambda}
\newcommand{\upara}{\msec}
\newcommand{\ubang}{\oc}

\newcommand{\size}[1]{\lvert #1 \rvert}
\newcommand{\unfold}[2]{\sharp^{#1}(#2)}

\newcommand{\lightll}{\textbf{LLL}}

\numberwithin{equation}{section}

\title{A Polynomial Time {\huge$\lambda$}-calculus \\ with
  Multithreading and Side Effects\thanks{Work partially supported by project
  ANR-08-BLANC-0211-01 ``COMPLICE'' and the Future and Emerging Technologies (FET)
programme within the Seventh Framework Programme for Research of the
European Commission, under FET-Open grant number: 243881 (project
CerCo).}}
\authorinfo{Antoine Madet}
{Univ Paris Diderot, Sorbonne Paris Cité\\
PPS, UMR 7126, CNRS, F-75205 Paris, France}
{madet@pps.univ-paris-diderot.fr}

\begin{document} 
\conferenceinfo{PPDP'12,} {September 19--21, 2012, Leuven, Belgium.}
\CopyrightYear{2012}
\copyrightdata{978-1-4503-1522-7/12/09}

\newtheorem{theorem}{Theorem}
\newtheorem{lemma}{Lemma}
\newtheorem{proposition}{Proposition}
\newtheorem{corollary}{Corollary}
\newtheorem{definition}{Definition}
\newtheorem{example}{Example}
\newtheorem{remark}{Remark}
\date{}
\maketitle 


\begin{abstract}
  The framework of \emph{light logics} has been extensively studied to
  control the complexity of higher-order functional programs. We
  propose an extension of this framework to multithreaded programs with side effects, focusing on the case of
  polynomial time. After introducing
   a modal $\lambda$-calculus with parallel composition
  and \emph{regions}, we prove that a
  realistic call-by-value evaluation strategy can be computed
  in polynomial time for a class of well-formed programs. The
  result relies on the simulation of call-by-value by a
  polynomial \emph{shallow-first} strategy which preserves the evaluation order of side effects.
 Then, we provide a polynomial type system that
  guarantees that well-typed programs do not go wrong. Finally, we
  illustrate the expressivity of the type system by giving
  a programming example of concurrent iteration producing side
    effects over an inductive
  data structure.
\end{abstract}


\section{Introduction}
Quantitative resource analysis of programs is a challenging task
in computer science. Besides being essential for the
development of safety-critical systems, it provides interesting
viewpoints on the structure of programs.

The framework of \emph{light logics} (see \emph{e.g.} \lightll~\cite{LLL},
\textbf{ELL}~\cite{Danos}, \textbf{SLL}~\cite{Lafont04})
which originates from Linear Logic~\cite{Girard87}, have been
deeply studied to control the complexity of higher-order
functional programs. In particular, 
polynomial time $\lambda$-calculi~\cite{Terui07,SoftLambda} have been
proposed as well as various
type systems~\cite{LMCS2008,CoppolaMartini2006} guaranteeing
complexity bounds of functional programs.
 Recently, Amadio and the author proposed an extension of the framework 
to a higher-order functional language with multithreading and side
effects~\cite{MadetA11}, focusing on the case of elementary
time (\textbf{ELL}).

In this paper, we consider a more reasonable complexity class: polynomial time. The
functional core of the language is the \emph{light}
$\lambda$-calculus~\cite{Terui07}  that features the
modalities \emph{bang} (written `$\oc$') and \emph{paragraph} (written
`$\msec$') of \lightll. The notion of \emph{depth} (the
number of nested modalities) which is standard in light logics is used to
control the duplication of data during the execution of programs. The
language is extended with side effects by
means of read and write operations on \emph{regions} which were
introduced to represent areas of the store~\cite{Lucassen88}. Threads
can be put in parallel and interact through a shared state.

There appears to be no direct combinatorial argument to bound a
call-by-value evaluation strategy by a polynomial.  However, 
the \emph{shallow-first} strategy (\emph{i.e.} redexes are eliminated
in a depth-increasing order) is known to be polynomial in the functional
case~\cite{LLL,Asperti98}. Using this result, Terui
shows~\cite{Terui07} that a class of \emph{well-formed} light $\lambda$-terms strongly
terminates in polynomial time (\emph{i.e.} every reduction strategy is polynomial) by
 proving that any reduction sequence can be
simulated by a \emph{longer} one which is shallow-first.
Following this method, our contribution is to show that a class of
well-formed call-by-value programs  with side effects and multithreading can
be simulated in polynomial time by shallow-first reductions.  The bound covers any scheduling policy and takes
 thread generation into account.

Reordering a reduction sequence into a shallow-first one is
non-trivial: the evaluation order of side effects must be kept unchanged in
order to preserve the semantics of the program. An additional difficulty is that
reordering produces non call-by-value sequences but fails for an
arbitrary larger relation (which may even require exponential
time). We identify an intermediate \emph{outer-bang} relation $\reducf$
which can be simulated by shallow-first ordering and this allows us to
simulate the call-by-value relation $\reduc_\sv$ which is
contained in the outer-bang relation. We illustrate this development in
Figure~\ref{summary}.
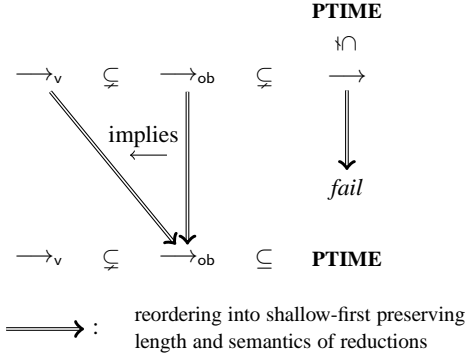
\begin{figure}[h]
\centering
  \begin{tikzpicture}                                                                          
    \matrix [matrix,column sep=0.3cm]                         
    {          
      &&&&  \node(expp){
        {\small \textbf{PTIME}} 
        };\\
      &&&&  \node(exppp){
        \rotatebox{-90}{$\subsetneq$}
        };\\
      \node(a){   $\reduc_\sv$  }; 
      &
      \node(q){    $\subsetneq$         };                                                                                    
      &
      \node(f){  $\reducf$ };                                                                            
      &
      \node(s){  $\subsetneq$ };      
      &                                                                      
      \node(v){  $\reduc$     };
      &                                                                      
      \node(sss){     };
&       \node(exp){  };\\[0.7cm]
\node(aa){};&&\node(bb){};&&\\[0.1cm]
&&&&\node(fail){\emph{fail}};
  \\[0.45cm]
        \node(asf){ $\reduc_\sv$ }; 
        &
       \node(u){  $\subsetneq$      };                                                                                    
       &
       \node(fsf){  $\reducf$ };                                                                            
       &
       \node(g){ $\subseteq$ };      
       &                                                                      
       \node(vsf){ {\small \textbf{PTIME}}       };\\
    };                                                                                        
     \draw [->,double] (a) to (fsf);                                                                 
     \draw [shorten <=1.1cm,shorten >=0.16cm,<-]  (aa) to node [above]
     {\quad\quad$\;\;$implies
     } (bb);                                                                 
     \draw [->,double] (f) to (fsf);                                                                 
     \draw [->,double] (v) to (fail);                                                                 
  \end{tikzpicture}         
  \begin{tikzpicture}                                                                          
    \matrix [matrix,column sep=0.05cm,row sep=1.5cm]                         
    {                                                                                       
       \node(blab){\phantom{:}}; 
       &[1cm] \node(blabb){:}; 
       & \node(blabbb){
         \begin{tabular}{ll}
{\small reordering into shallow-first preserving}\\
{\small length and semantics of reductions}
         \end{tabular}};\\
    };                                                                                   
     \draw [->,double] (blab) to (blabb);                                   
  \end{tikzpicture}  
  \caption{Simulation by shallow-first ordering}
  \label{summary}
\end{figure}

The paper is organized as follows.
We start by presenting the language with
multithreading and regions in Section~\ref{sec-language} and define
the largest reduction relation. Then, we introduce a \emph{polynomial depth
  system} in Section~\ref{sec-depth} to control the depth of
program occurrences. Well-formed programs in the depth system  
  follow  Terui's discipline~\cite{Terui07} on the functional side and the \emph{stratification of
  regions} by depth level that we introduced
previously~\cite{MadetA11}. 
We prove in Section~\ref{sec-simulation} that the class of outer-bang strategies
(containing call-by-value) can be simulated by
shallow-first reductions of exactly the same length.
We review the proof of polynomial soundness of the shallow-first
strategy in
Section~\ref{sec-soundness}.
We provide a \emph{polynomial type system} in Section~\ref{sec-types}
which results from a simple decoration of the polynomial depth system
with linear types. We derive the standard subject reduction
proposition and progress proposition which states that well-types
programs reduce to values.  Finally, we
illustrate the expressivity of the type system
in Section~\ref{sec-expr} by showing that it
is polynomially complete in the extensional sense and we give a
programming example of a concurrent iteration producing side
    effects over an inductive data structure.

\section{A modal $\lambda$-calculus\\ with multithreading and regions}
\label{sec-language}
As mentioned previously, the functional core of the language is a
modal $\lambda$-calculus with constructors and destructors for
the modalities `$\oc$' and `$\msec$' that are used to control the
duplication of data. The global store is partitioned into a finite number of
 regions where each region abstracts a set of
 memory locations. Following~\cite{Amadio09}, side effects are
 produced by read and write operators on regions. A parallel operator
allows to evaluate concurrently several terms which can communicate
through regions.  As we shall see in Section~\ref{sec-expr}, this abstract non-deterministic
language entails complexity bounds for languages with concrete
memory locations  representing \emph{e.g.} references, channels or
signals.
 
The syntax of the language is presented in Figure~\ref{language-syntax}.
\begin{figure}[ht] 
$$
\begin{array}{l@{\;\;}rcl}
\textrm{-variables}&\multicolumn{3}{l}{x,y,\ldots}\\
\textrm{-regions}&\multicolumn{3}{l}{r,r',\ldots} \\
\textrm{-terms} & M &::=& x \mid r \mid \star \mid \lambda x. M \mid MM \mid \bang M \mid
\msec M \\ 
&&&\letb{x}{M}{M} \mid \letp{x}{M}{M}\\ 
&&&\get{r} \mid \st{r}{M} \mid (M \para M)\\
\textrm{-stores} &S &::=& \store{r}{M} \mid (S \para S)\\
\textrm{-programs} & P &::=& M \mid S \mid (P \para P)
\end{array}
$$
\caption{Syntax of the language}
\label{language-syntax}
\end{figure}
We have the usual set of variables $x,y,\ldots$ and a set of regions
$r,r',\ldots$ The set of terms $M$ contains variables, regions, the
terminal value (unit) $\star$, $\lambda$-abstractions, applications,
modal terms $\oc M$ and $\msec M$ (resp. called $\oc$-terms and
$\msec$-terms) and the associated $\mathsf{let}\,\oc$-binders and
$\mathsf{let}\,\msec$-binders. We have an operator $\get{r}$
to read a region $r$, an operator $\st{r}{M}$ to assign a term $M$ to
a region $r$ and a parallel operator $(M \para N)$ to evaluate $M$ and
$N$ in parallel. A store $S$ is the composition of several assignments
$\store{r}{M}$ in parallel and a program $P$ is the combination of
several terms and stores in parallel. Note that stores are global, \emph{i.e.}
they always occur in empty contexts.

In the following we write $\dagger$ for $\dagger \in
\set{\oc,\msec}$ and we define $\dagger^0 M = M$ and $\dagger^{n+1} M
= \dagger(\dagger^n M)$. Terms $\lambda x.M$ and $\letd{x}{N}{M}$ bind
occurrences of $x$ in $M$. The set of free variables of $M$ is denoted
by $\fv{M}$. The number of free occurrences of $x$ in $M$ is denoted
by $\fo{x}{M}$. The number of free occurrences in $M$ is denoted by
$\foa{M}$. $\sub{M}{N/x}$ denotes the term $M$ in which each free
occurrence of $x$ has been substituted by $N$.

Each program has an \emph{abstract syntax tree} where variables,
regions and unit constants are
 leaves, $\lambda$-abstractions and $\dagger$-terms have one child, and
 applications and $\mathsf{let}\,\dagger$-binders have two children. An
 example is given in Figure~\ref{trees}.
\begin{figure}[ht]
  \centering
 $$P = \letb{x}{\get{r}}{\st{r}{(\oc x)(\msec x)}} \para
        \store{r}{\oc(\lambda x.x\star)}$$
    \begin{tikzpicture}[font=\scriptsize,level distance=0.65cm]
      \node {$\para^{\epsilon}$}[sibling distance=3cm]
      child {node {$\letbang x^0$}[sibling distance=1cm]
        child {node {$\get{r}^{00}$}}
        child {node {$\mathsf{set}(r)^{01}$}
          child {node {$@^{010}$}
            child {node {$\oc^{0100}$}
              child {node {$x^{01000}$}}}
            child {node {$\msec^{0101}$}
              child {node {$x^{01010}$}}}}}}
      child {node {$r \Leftarrow^{1}$}
        child {node {$ \oc^{10}$}
          child {node {$ \lambda x^{100}$}
            child {node {$@^{1000}$}[sibling distance=1cm]
              child {node {$x^{10000}$}}
              child {node {$\star^{10001}$}}}}}};
    \end{tikzpicture}
\caption{Syntax tree and addresses of $P$}
  \label{trees}
\end{figure}
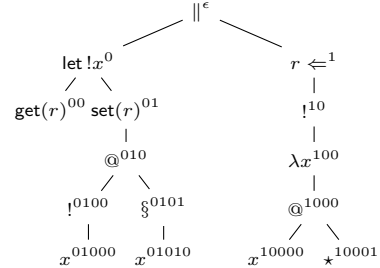
A path starting from the root to a node of the tree denotes an
\emph{occurrence} of the program whose address is a word $w \in
\set{0,1}^{*}$ hereby denoted in exponent form. We write $w
\sqsubseteq w'$ when $w$ is a prefix of $w'$. We denote the number of
occurrences in $P$ by $\size{P}$.

The operational semantics of the language is given in
Figure~\ref{semantics}. In order to prove the later simulation result,  the largest
reduction relation $\reduc$ (which shall contain call-by-value) is presented. 
\begin{figure}[ht]
$$
\begin{array}{rcl}
\multicolumn{3}{c}{\textrm{-structural rules-}}\\
 P\para P' &\equiv &P'\para P                                \\
 (P\para P')\para P'' &\equiv &P \para (P' \para P'')          
\end{array}
$$
$$
  \begin{array}{ r c l}
\multicolumn{3}{c}{\textrm{-evaluation contexts-}}\\
  E&::=&[\cdot]\mid \lambda x.E \mid EM \mid ME \mid \bang E \mid \mathsection E \\
&&  \letb{x}{E}{M} \mid \letp{x}{E}{M}\\
&&  \letb{x}{M}{E} \mid \letp{x}{M}{E}\\
&& \st{r}{E} \mid \store{r}{E} \mid (E\para P) \mid (P\para E)
\end{array}
$$
$$
  \begin{array}{@{}lr @{\;\;}c@{\;\;} l@{}}
\multicolumn{4}{c}{\textrm{-reduction rules-}}\\
    (\beta) &E[(\lambda x.M)N] &\reduc& E[\sub{M}{N/x}]\\
    (\oc) &E[\letm{x}{\oc N}{M}] &\reduc& E[\sub{M}{N/x}]\\
    (\msec) &E[\letp{x}{\msec N}{M}] &\reduc& E[\sub{M}{N/x}]\\
     (\getsf) &E[\get{r}] \para \store{r}{M} &\reduc& E[M] \\
     (\setsf) & E[\st{r}{M}] &\reduc& E[\star] \para \store{r}{M}
     \text{{\scriptsize--if $\fv{M} = \emptyset$}}\\
     (\gcsf) & E[\star \para  M] &\reduc& E[M]
  \end{array}
$$
  \caption{Operational semantics}
  \label{semantics}
\end{figure}

Programs are considered up to a structural equivalence $\equiv$ which
contains the equations for $\alpha$-renaming, commutativity and
associativity of parallel composition. Reduction rules apply modulo
structural equivalence, in an evaluation context $E$ which can be any
program with exactly one occurrence of a special variable `$[\cdot]$',
called the \emph{hole}. We write $E[M]$ for $\sub{E}{M/[\cdot]}$.
Each rule is identified by its name. $(\beta)$ is the usual
$\beta$-reduction. $(\dagger)$ are rules for filtering modal
terms. (\getsf) is for \emph{consuming} a term from a region. (\setsf)
is for \emph{assigning} a closed term to a
region. (\gcsf) is for \emph{erasing} a terminated thread.

First, note that the reduction rule (\setsf) generates a \emph{global} assignment, that
is out of the evaluation context $E$. In turn, we require
$M$ to be closed such that it does not contain variables bound in $E$.
Second,  several terms can be assigned to a single region. This cumulative semantics allows the
simulation of several memory locations by a single region. In turn,
reading a region consists in consuming \emph{non-deterministically} one of
the assigned terms. 

The reduction is very `liberal' with side effects. The contexts
$(P \para E)$ and $(E \para P)$ embed any scheduling of threads. Moreover, 
 contexts of the shape $\store{r}{E}$ allow evaluation in the store as
 exemplified in the following possible reduction:
 \begin{equation*}
 \begin{array}{@{}r@{\;\;}c@{\;\;}l@{}}
 \st{r}{\lambda x.\get{r}}\para \store{r}{M} &\reduc&  \star \para
 \store{r}{\lambda x.\get{r}} \para \store{r}{M}\\
 & \reduc& \star \para \store{r}{\lambda x.M}  
 \end{array}   
 \end{equation*}

In the rules $(\beta),(\dagger),(\gcsf)$, the \emph{redex} denotes the term inside the context of
the left hand-side  and the \emph{contractum}
denotes the term inside the context of the right hand-side. In the
rule $(\getsf)$, the redex is $\get{r}$ and the contractum is $M$. In
the rule $(\setsf)$, the redex is $\st{r}{M}$ and the contractum is $M$.
 Finally, $\preduc$ denotes the
transitive closure of $\reduc$ and $\treduc$ denotes the
reflexive closure of $\preduc$.

\section{A polynomial depth system}
\label{sec-depth}
In this section, we first review 
the principles of well-formed light $\lambda$-terms
(Subsection~\ref{light}) and then the stratification of regions by depth level
(Subsection~\ref{strat}). Eventually we combine the two as a set of inference rules that
 characterizes a class of \emph{well-formed} programs (Subsection~\ref{rules}). 
 
\subsection{On light $\lambda$-terms}
\label{light}

First, we define the notion of depth.
\begin{definition}\label{depth-functional}
  The \emph{depth} $d(w)$ of an occurrence $w$ in a program $P$ is the
  number of $\dagger$ labels that the path leading to the end node
  crosses.  The depth $d(P)$ of program $P$ is the maximum depth of
  its occurrences.
\end{definition}
With reference to Figure~\ref{trees}, $ d(01000) = d(01010) = d(100) =
d(1000) = d(10000) = d(10001) = 1$, whereas other occurrences have
depth $0$. In particular, $d(0100) = d(0101) = d(10)=0$; what matters
in computing the depth of an occurrence is the number of $\dagger$'s
that precede strictly the end node. Thus $d(P) = 1$.  In the sequel,
we say that a program \emph{occurs} at depth $i$ when it corresponds to
an occurrence of
depth $i$. For example, $\get{r}$ occur at depth 0 in
$P$. We write $\reduci{i}$ when the redex occurs at depth $i$; we
write $\size{P}_i$ for the number of occurrences at depth $i$ of $P$.

Then we can define \emph{shallow-first} reductions. 
\begin{definition}
\label{def-std}
A \emph{shallow-first reduction sequence} $P_1 \reduci{i_1} P_2 \reduci{i_2} \ldots 
  \reduci{i_n} P_n$ is  such that  $m < n$ implies $i_m \leq i_n$. A
  \emph{shallow-first strategy} is a strategy that produces shallow-first
   sequences.
\end{definition}
The polynomial soundness of shallow-first strategies relies on the
following properties: when $P \treduci{i} P'$,
\begin{align}
 \label{o1} d(P') &\leq d(P)\\
 \label{o2} \size{P'}_j &\leq \size{P}_j \text{ for } j < i\\
 \label{o3} \size{P'}_i &< \size{P}_i\\
 \label{o4} \size{P'} &\leq \size{P}^2
\end{align}
To see this in a simple way, assume $P$ is a program such that $d(P) = 2$. By 
properties~\eqref{o1},\eqref{o2},\eqref{o3} we can
eliminate all the redexes of $P$ with the shallow-first sequence $P \treduci{0} P'
\treduci{1} P'' \treduci{2} P'''$. By  property~\eqref{o4},
 $\size{P'''} \leq \size{P}^8$. By
properties~\eqref{o3} the length $l$ of the sequence is
such that $l \leq \size{P}
+ \size{P'} + \size{P''} = p$. Since we can show that $p \leq \size{P}^8$ we
conclude that the shallow-first evaluation of $P$ can be computed in polynomial time.

The well-formedness criterions of light $\lambda$-terms are intended
to ensure the above four properties. These criterions can be summarized as follows:
\begin{itemize}
\item $\lambda$-abstraction is affine: in $\lambda x.M$, $x$ may
  occur at most once and at depth $0$ in $M$.
\item $\mathsf{let}\,\oc$-binders are for duplication: in
  $\letb{x}{M}{N}$, $x$ may occur arbitrarily many times and at depth
  $1$ in $N$.
\item $\mathsf{let}\,\msec$-binders are affine: in
  $\letp{x}{M}{N}$, $x$ may occur at most once and at depth $1$ in
  $N$. The depth of $x$ must be due to a $\msec$ modality.
\item a $\oc$-term may contain at most one occurrence of free variable, whereas
  a $\msec$-term can contain many occurrences of free variables.
\end{itemize}
By the first three criterions, we observe the following. The depth of
a term never increases (property~\eqref{o1}) since the reduction rules
$(\beta)$,$(\oc)$ and $(\msec)$ substitute a term for a variable
occurring at the same depth. Reduction rules $(\beta)$ and $(\msec)$
are strictly size-decreasing since the corresponding binders are
affine. A reduction $(\oc$) is strictly size-decreasing at the depth
where the redex occurs but potentially size-increasing at deeper
levels. Therefore properties~\eqref{o2} and~\eqref{o3} are also
guaranteed.  The fourth criterion is intended to ensure
a quadratic size increase (property~\eqref{o4}). Indeed, take the term $Z$ borrowed
from~\cite{Terui07} that respects the first three criterions but not the
fourth:
\begin{equation}
  \label{blowup}
  \begin{split}
Z = \lambda x.\letb{x}{x}{\oc(xx)}\\
\underset{n \text{ times}}{\underbrace{Z\ldots (Z(Z}} \oc y)) \treduc
\oc(\underset{2^n \text{ times}}{\underbrace{yy\ldots y}})
  \end{split}
\end{equation}
It may trigger  an exponential size explosion by repeated application
of the duplicating rule $(\oc)$. The
following term
\begin{align}
\label{blocked}
\begin{split}
        &Y = \lambda x.\letb{x}{x}{\msec(xx)}\\
&\underset{n \text{ times}}{\underbrace{Y\ldots (Y(Y}} \oc y))\\
 \treduc\;&
\underset{n-2 \text{ times}}{\underbrace{Y\ldots (Y
    (Y}}(\letb{x}{\msec(yy)}{\msec(xx)}))) \nrightarrow
\end{split}
\end{align}
respects the four criterions but cannot be used to apply $(\oc)$ exponentially.

\subsection{On the stratification of regions by depth}
\label{strat}
In our
previous work on elementary time~\cite{MadetA11}, we
analyzed the impact of side effects on the depth of occurrences and
remarked that arbitrary reads and writes could increase the depth of 
programs. In the reduction sequence 
\begin{equation}
  \label{bad1}
  \begin{split}
  (\lambda x.\st{r}{x}\para \msec \get{r})\oc M \treduc&\; \msec
  \get{r} \para \store{r}{\oc M}\\ \reduc\;& \msec\oc M    
  \end{split}
\end{equation}
the occurrence $M$ moves from depth $1$ to depth $2$ during the last
reduction step, because the read occurs at depth $0$ while the write
occurs at depth $1$. 

Following this analysis, we introduced \emph{region contexts} in order
to constrain the depth at which side effects occur. A region context
$$
R = r_1:\delta_1,\ldots,r_n:\delta_n
$$
associates a natural number $\delta_i$ to each region $r_i$ in a
finite set of regions $\set{r_1,\ldots,r_n}$ that we write
$dom(R)$. We write $R(r_i)$ for $\delta_i$. Then, the rules of the
elementary depth system were designed in such a way that $\get{r_i}$ and $\st{r_i}{M}$
may only occur at depth $\delta_i$, thus rejecting~\eqref{bad1}.

Moreover, we remarked that since stores are global, that is they always
occur at depth $0$, assigning a term to a region breaks
stratification whenever $\delta_i > 0$. Indeed, in the reduction
\begin{equation}
  \label{bad2}
  \msec \st{r}{M} \reduc \msec\star \para \store{r}{M}  
\end{equation}
where $R(r)$ should be $1$, the occurrence $M$ moves from depth $1$ to depth $0$. Therefore, we
revised the definition of depth as follows.
\begin{definition}
\label{revised-depth}
Let $P$ be a program and $R$ a region context where $dom(R)$ contains all
the regions of $P$.
The \emph{revised depth} $d(w)$ of an occurrence $w$ of $P$ is the number of
$\dagger$ labels that the path leading to the end node crosses, plus
$R(r)$ if the path crosses a store label $r \Leftarrow$. The revised
depth
$d(P)$ of a program $P$ is the maximum revised depth of its occurrences.
\end{definition}
By considering this revised definition of depth, in~\eqref{bad2} the occurrence $M$ stays at 
depth $1$.  In Figure~\ref{trees} we now get $d(01000) = d(01010) = 1$, $d(10) = R(r)$
and $d(100) = d(1000) = d(10000) = d(10001) = R(r) + 1$. Other
occurrences have depth $0$.  From now on we shall say depth for 
 the revised definition of depth.

\subsection{Inference rules}
\label{rules}
Now we introduce the inference rules of the polynomial depth system.
First, we define region contexts $R$  and variable contexts $\Gamma$
as follows:
$$
\begin{array}{rcl}
R &=& r_1:\delta_1,\ldots,r_n:\delta_n\\
\Gamma &=& x_{1}:u_{1},\ldots,x_{n}:u_{n}
\end{array}
$$
Regions contexts are described in the previous subsection. A variable context
associates each variable with a usage $u \in
\set{\uaff,\upara,\ubang}$ which  constrains the variable to be bound
by a $\lambda$-abstraction, a $\mathsf{let}\,\msec$-binder or a
$\mathsf{let}\,\oc$-binder respectively. We write
$\Gamma_u$ if $dom(\Gamma)$ only contains variables with usage $u$.
A depth judgement has the shape
$$R;\Gamma \vdash^\delta P$$
where $\delta$ is a natural number. It  should entail the following:
 \begin{itemize}
 \item if $x:\lambda \in \Gamma$ then
   $x$ occurs at depth $\delta$ in $\dagger^\delta P$,
 \item if $x:\dagger \in \Gamma$ then
   $x$ occurs at depth $\delta+1$ in $\dagger^\delta P$,
\item  if $r:\delta' \in R$ then $\get{r}$/$\mathsf{set}(r)$ occur at depth
$\delta'$ in $\dagger^\delta P$.
 \end{itemize}

The inference rules of the  depth system are presented in
Figure~\ref{depth-system-poly}.
\begin{figure}[ht]
$$
\begin{array}{c}
  \inference
  {x:\uaff \in \Gamma }
  {R;\Gamma  \vdash^\delta x}
\qquad
\inference
{}
{R;\Gamma \vdash^\delta \star} 

\qquad

\inference
{}
{R;\Gamma  \vdash^\delta r}\\ \\ 

  \inference
  {  \fo{x}{M} = 1 
    \\ R;\Gamma,x:\uaff  \vdash^\delta M}
  {R;\Gamma \vdash^\delta \lambda x.M}

  \qquad

  \inference
  {R;\Gamma \vdash^\delta M & R;\Gamma \vdash^\delta
    N}
  {R;\Gamma \vdash^\delta MN}\\ \\

  \inference
  { \foa{M} \leq 1 \\ R;\Gamma_{\uaff} \vdash^{\delta+1} M}
  {R; \Gamma_{\ubang},\Delta_{\upara}, \Psi_{\uaff} \vdash^\delta \oc M}

  \quad

  \inference
  { \fo{x}{N} \geq 1&R;\Gamma \vdash^\delta M \\  R;\Gamma,x:\ubang \vdash^\delta
    N}
  {R;\Gamma \vdash^\delta \letb{x}{M}{N}}\\ \\

  \inference
  { R;\Gamma_{\uaff},\Delta_{\uaff} \vdash^{\delta+1} M}
  { R;\Gamma_{\ubang},\Delta_{\upara},\Psi_{\uaff} \vdash^\delta \msec M}

  \;\;\;

  \inference
  { \fo{x}{N} = 1&R;\Gamma \vdash^\delta M \\ R;\Gamma,x:\upara \vdash^\delta
    N }
  {R;\Gamma \vdash^\delta \letp{x}{M}{N}}\\ \\

\inference
{r:\delta \in R}
{R;\Gamma  \vdash^\delta \get{r}}
\qquad

\inference
{r:\delta \in R & R;\Gamma  \vdash^\delta
M}
{R;\Gamma  \vdash^\delta \st{r}{M}} 

\\\\

\inference
{r:\delta \in R & R;\Gamma  \vdash^\delta M}
{R;\Gamma \vdash^0  \store{r}{M}}

\qquad

\inference
{i=1,2&
  R;\Gamma  \vdash^\delta P_i}
{R;\Gamma  \vdash^\delta (P_1\para P_2)}  
\end{array}
$$
\caption{A polynomial depth system}
\label{depth-system-poly}
\end{figure}
We comment on the handling of usages. Variables are introduced with usage
$\lambda$. The construction of  $\oc$-terms updates the usage of
variables to $\oc$ if they all previously had usage $\lambda$. The construction of $\msec$-terms updates the usage of
variables to $\msec$ for one part and $\oc$ for the other part if they
all previously had usage $\lambda$. In both constructions, 
contexts with other usages can be weakened. As a result,
$\lambda$-abstractions bind variables occurring at depth $0$,
 $\mathsf{let}\,\oc$-binders bind variables occurring at depth $1$ in 
$\oc$-terms or $\msec$-terms, and  $\mathsf{let}\,\msec$-binders bind variables occurring at depth $1$ in
 $\msec$-terms.

To control the duplication of data, the rules for binders have predicates which specify how many
occurrences can be bound. $\lambda$-abstractions
and $\mathsf{let}\,\msec$-binders are linear by  predicate
$\fo{x}{M} = 1$ and $\mathsf{let}\,\oc$-binders are at least linear by
 predicate $\fo{x}{M} \geq 1$. 

The depth $\delta$ of the judgement is decremented when constructing
$\dagger$-terms. This allows to stratify regions by depth level by requiring
that $\delta = R(r)$ in the rules for $\get{r}$ and $\st{r}{M}$. A store assignment
$\store{r}{M}$ is global hence its judgement has depth $0$ whereas the premise
has depth $R(r)$ (this reflects the revised
notion of depth).


\begin{definition}\textit{(Well-formedness)}
\label{def-well-formed}
A program $P$ is \emph{well-formed} if 
a judgement $R;\Gamma \vdash^\delta P$ can be derived for some $R$, $\Gamma$ and $\delta$.
\end{definition}

\begin{example}
The program $P$ of Figure~\ref{trees} is
        well-formed by composition of the two derivation trees of Figure~\ref{derivtrees}.
\begin{figure*}[ht]
  \centering
  \begin{tabular}{@{}c@{}c@{}}
    \begin{minipage}{0.63\linewidth}
  \begin{prooftree}
  \AxiomC{}
  \UnaryInfC{$r:0;- \vdash^0 r$}
  \UnaryInfC{$r:0;- \vdash^0 \get{r}$}
  \AxiomC{}
  \UnaryInfC{$r:0;x:\ubang \vdash^0 r$}
  \AxiomC{}
  \UnaryInfC{$r:0;x:\lambda\vdash^1 x $}
  \UnaryInfC{$r:0;x:\ubang\vdash^0 \oc x$}
  \AxiomC{}
  \UnaryInfC{$r:0;x:\lambda\vdash^1 x $}
  \UnaryInfC{$r:0;x:\ubang\vdash^0 \msec x $}
  \BinaryInfC{$r:0;x:\ubang\vdash^0 \oc x \msec x$}
  \BinaryInfC{$r:0; x:\ubang \vdash^0 \st{r}{\oc x \msec x}$}
  \BinaryInfC{$r:0;- \vdash^0 \letb{x}{\get{r}}{\st{r}{\oc
        x\msec x}}$}
\end{prooftree}      
    \end{minipage}
 &
     \begin{minipage}{0.37\linewidth}
\begin{prooftree}
    \AxiomC{}
   \UnaryInfC{$r:0;x:\lambda \vdash^{1} x$}
    \AxiomC{}
   \UnaryInfC{$r:0;x:\lambda \vdash^{1} \star$}
   \BinaryInfC{$r:0;x:\lambda \vdash^{1} x\star$}
   \UnaryInfC{$r:0;- \vdash^{1} \lambda x.x\star$}
   \UnaryInfC{$r:0;- \vdash^0 \oc (\lambda x.x\star)$}
  \UnaryInfC{$r:0;- \vdash^0 \store{r}{\oc (\lambda x.x\star)}$}
\end{prooftree}    
     \end{minipage}

  \end{tabular}

  \caption{Derivation trees}
  \label{derivtrees}
\end{figure*}
The program $Z$ given in~\eqref{blowup} is not well-formed.
\end{example}

The depth system is strictly linear in the sense that it is not
possible to bind $0$ occurrences. We shall see in
Section~\ref{sec-simulation} that it allows for a major simplification
of the proof of simulation.  However, this impossibility to discard data is a
notable restriction over light $\lambda$-terms. In a call-by-value
setting, the sequential composition $M;N$ is usually encoded as the
non well-formed term $(\lambda z.N)M$ where $z \notin \fv{N}$ is used
to discard the terminal value of $M$. We show that side effects can be
used to simulate the discarding of data even though the depth system
is strictly linear. Assume that we dispose of a specific region $gr$
collecting `garbage' values at each depth level of a program. Then
$M;N$ could be encoded as the well-formed program $(\lambda
z.\st{gr}{z}\para N)M$. Using a call-by-value semantics, we would
observe the following reduction sequence
\begin{align*}
M;N \treduc V;N \reduc \st{gr}{V} \para N &\reduc \star \para N \para
\store{gr}{V}\\ 
&\reduc N \para \store{gr}{V}
\end{align*}
where $\star$ has been erased by (\gcsf) and $V$ has been garbage collected into $gr$.

Finally we derive the following lemmas on the depth system in order to
get the subject reduction proposition.
\begin{lemma}[Weakening and Substitution]\hfill
  \label{substitution}
  \begin{enumerate}
  \item   \label{weakening}   If $R;\Gamma \vdash^\delta P$ then $R;\Gamma,\Gamma'\vdash^\delta P$.
  \item\label{sub1}   If $R;\Gamma,x:\uaff \vdash^{\delta} M$ and
  $R;\Gamma \vdash^\delta N$\\ then $R;\Gamma \vdash^\delta
  \sub{M}{N/x}$.
  \item\label{sub2}   If $R;\Gamma,x:\upara \vdash^{\delta} M$ and
  $R;\Gamma \vdash^\delta \msec N$\\ then $R;\Gamma \vdash^\delta
  \sub{M}{N/x}$.
  \item\label{sub3}   If $R;\Gamma,x:\ubang \vdash^{\delta} M$ and
  $R;\Gamma \vdash^\delta \oc N $\\ then $R;\Gamma \vdash^\delta
  \sub{M}{N/x}$.
  \end{enumerate}
\end{lemma}
\begin{proposition}[Subject reduction]
\label{subred} If $R;\Gamma
  \vdash^\delta P$ and $P\reduc P'$ then $R;\Gamma
  \vdash^\delta P'$ and $d(P) \geq d(P')$.
\end{proposition}

\section{Simulation by shallow-first}
\label{sec-simulation}
In this section, we first explain why we need a class of \emph{outer-bang}
 reduction strategies (Subsection~\ref{secouter}). Then, we prove that shallow-first
simulates any outer-bang strategy and that the result applies to call-by-value (Subsection~\ref{sim-sf}).

\subsection{Towards outer-bang strategies}
\label{secouter}
Reordering a reduction sequence into a shallow-first one is an iterating
process where each iteration consists in commuting two consecutive
reduction steps which are applied in `deep-first' order. 

First, we show that this process requires a reduction  which is strictly larger
than an usual call-by-value relation. Informally, assume $\dagger V$ denotes a
value. The following two reduction steps in call-by-value style
    \begin{equation*}
      \begin{split}
      \st{r}{\dagger M} \reduci{1}\st{r}{\dagger V}
      \reduci{0}   \star \para \store{r}{\dagger V}        
      \end{split}
    \end{equation*}
commute into the shallow-first sequence
    \begin{equation*}
      \begin{split}
      \st{r}{\dagger M} \reduci{0} \star \para \store{r}{\dagger M}
      \reduci{1}   \star \para \store{r}{\dagger V}        
      \end{split}
    \end{equation*}
which is obviously not  call-by-value: first, we write a non-value
$\dagger M$ to the store and second we reduce \emph{in} the store! As
another example, the following two reduction steps in
call-by-value style
\begin{equation*}
(\lambda x.\lambda y.xy)\!\dagger\!M \reduci{1} (\lambda x.\lambda y.xy)\!\dagger\! V
\reduci{0} \lambda y.(\dagger V)y  
\end{equation*}
commute into the shallow-first sequence
\begin{equation*}
(\lambda x.\lambda y.xy)\!\dagger\! M \reduci{0} \lambda y.(\dagger M)y \reduci{i}
\lambda y.(\dagger V)y  
\end{equation*}
which is  not call-by-value: we need to reduce inside a
$\lambda$-abstraction and this is not compatible with the usual notion of value.

Second, we show that an arbitrary relation like $\reduc$ is too large to be simulated
by shallow-first sequences. For instance, consider the following
reduction of a well-formed program:
\begin{equation}
  \label{bangreduc}
  \begin{split}
    &\letb{x}{\oc \get{r}}{\msec(xx)} \para \store{r}{M}\\
    \reduci{1}\;&\letb{x}{\oc M}{\msec(xx)}\\ 
    \reduci{0}\;& \msec(MM)  
  \end{split}
\end{equation}
This sequence is deep-first; it can be reordered into a
shallow-first one as follows:
\begin{equation}
\label{dupget}
\begin{split}
  &\letb{x}{\oc \get{r}}{\msec(xx)} \para \store{r}{M}\\
  \reduci{0}\;& \msec(\get{r}\get{r}) \para \store{r}{M}\\
  \reduci{1}\;& \msec(M\get{r}) \nrightarrow
\end{split}
\end{equation}
However, the sequence cannot be confluent with the previous one for we
try to read the region two times by duplicating the redex $\get{r}$. It turns out that a non shallow-first strategy
may require exponential time in the presence of side effects. Consider
the well-formed
$\lambda$-abstraction
\begin{equation*}
F = \lambda  x.\letp{x}{x}{\msec\st{r}{x} ; \oc \get{r}}  
\end{equation*}
which transforms a $\msec$-term into a $\oc$-term (think of the type
$\msec A \multimap \oc A$ that would be rejected in \lightll). Then, 
building on program $Z$ given in~\eqref{blowup}, take
\begin{equation*}
Z' = \lambda x.\letb{x}{x}{F\msec{(xx)}}
\end{equation*}
We observe an exponential explosion of the size of the
following well-formed 
program:
\begin{equation*}
\begin{split}
  &\underset{n \text{ times}}{\underbrace{Z'Z'\ldots Z'}} \oc \star\\ 
  \treduc\;&  \underset{n-1 \text{ times}}{\underbrace{Z'Z'\ldots
      Z'}}(F\msec{(\star\star)})\\
  \treduc\;& \underset{n-1 \text{
      times}}{\underbrace{Z'Z'\ldots Z'}}(\oc{(\star\star)}) \para
  \store{gr}{\msec\star}\\
  \treduc\;&
  \underset{ 2^n \text{ times}}{\oc(\underbrace{\star\star\ldots
      \star})} \para 
  \underset{n \text{ times}}{\underbrace{\store{gr}{\msec\star} \para \ldots \para \store{gr}{\msec\star}}}
\end{split}
\end{equation*}
where $gr$ is a region collecting the garbage produced by the sequential
composition operator of $F$. This previous sequence is not
shallow-first since the redexes $\st{r}{M}$ and $\get{r}$ occurring at depth
$1$ are alternatively applied with other redexes occurring at depth $0$. A shallow-first strategy would produce the
 reduction sequence
\begin{equation*}
\underset{n \text{ times}}{\underbrace{Z'Z'\ldots Z'}} \oc \star
\treduc \oc(\star\star\underset{n-1 \text{
    times}}{\underbrace{\get{r}\get{r}\ldots\get{r}}}) \para S
\end{equation*}
where $S$ is the same garbage store as previously but we observe no
size explosion.

Following these observations, our contribution is to identify an intermediate \emph{outer-bang} reduction relation
 that can be simulated by shallow-first sequences. The keypoint is to
 prevent reductions inside $\oc$-terms  like in
sequence~\eqref{bangreduc}.
For this, we define the \emph{outer-bang} evaluation
  contexts $F$ in Figure~\ref{linear-ctxts}.
\begin{figure}[ht]
\begin{equation*}
\begin{array}{rcl}
F &::=& [\cdot] \mid \lambda x.F \mid FM \mid MF \mid \msec F\\
&&
\letd{x}{F}{M} \mid \letd{x}{M}{F}\\
  &&\st{r}{F} \mid (F \para M)
\mid (M \para F) \mid  \store{r}{F}  
\end{array}
\end{equation*}
\caption{Outer-bang evaluation contexts}
\label{linear-ctxts}
\end{figure}
They are not decomposable in a
context of the shape $E[\oc E']$ and thus cannot be used to reduce in $\oc$-terms. In the sequel, $\reducf$
denotes reduction modulo evaluation contexts $F$.
\subsection{Simulation of outer-bang strategies}
\label{sim-sf}
After identifying a proper outer-bang relation $\reducf$,
the main difficulty is to preserve the
evaluation order of side effects by shallow-first reordering. For
example, the following two reduction steps do not commute:
\begin{equation}
  \label{separated}
  \begin{split}
      &F_1[\st{r}{Q}] \para F_2[\get{r}]\\
      \reduci{i}\;& F_1[\star] \para
  F_2[\get{r}] \para \store{r}{Q}\\
  \reduci{j}\;& F_1[\star] \para F_2[Q]
  \end{split}
\end{equation}
We claim that this is not an issue since the depth system
enforces that side effects on a given region can only occur at fixed depth,
hence that $i = j$. Therefore, we should never need to `swap' a read
with a write on the same region.

We can prove the following crucial lemma.
\begin{lemma}[Swapping]
  \label{swap}
  Let $P$ be a well-formed program such that \mbox{$P \reducif{i} P_1 \reducif{j} P_2$}
  and $i > j$. Then, there exists $P'$ such that \mbox{$P \reducif{j}P' \reducif{i} P_2$}.
\end{lemma}
\begin{proof}
  We write $M$  the contractum of the reduction $P \reducif{i} P_1$
  and $N$ the redex of the reduction $P_1 \reducif{j} P_2$. Assume they
  occur at addresses $w_m$ and $w_n$
   in $P_1$. We distinguish three cases: (1)
   $M$ and $N$ are separated (neither $w_m \sqsubseteq w_n$ nor $w_m
   \sqsupseteq w_n$); (2) $M$ contains $N$ ($w_m \sqsubseteq w_n$); (3) $N$ strictly contains $M$ ($w_m \sqsupseteq w_n$ and $w_m \neq w_n$).
   For each of them we discuss a crucial subcase:
  \begin{enumerate}
  \item
    Assume $M$ is the contractum of a (\setsf) rule and that $N$
    is the redex of a (\getsf) rule related to the same region. This
    case has been introduced in example ~\eqref{separated} where $M$
    and $N$ are separated by a parallel node.
    By well-formedness of $P$, the redexes $\get{r}$ and $\st{r}{Q}$
    must occur at the same depth, that is $i=j$, and we conclude that
    we do not need to swap the reductions.
  \item If the contractum $M$ contains the redex $N$, $N$
    may not exist yet in $P$ which makes the swapping impossible.
    We remark  that, for any well-formed
    program $Q$ such that $Q \reducif{d} Q'$, both the redex and the
    contractum  occur at depth $d$. In particular,
    this is true when a contractum occurs in the store as follows: $$Q =
    F[\st{r}{T}] \reducif{d} Q'= F[\star] \para \store{r}{T}$$ By
    well-formedness of $Q$, there exists a region context $R$ such
    that $R(r)=d$ and the redex $\st{r}{T}$ occurs at depth $d$. By
    the revised definition of depth, the contractum $T$ occurs at depth
    $d$ in the store. As a result of this remark, $M$ occurs
    at depth $i$ and $N$ occurs at depth $j$. Since $i > j$, it is
    clear that the contractum $M$ cannot contain the redex $N$ and this
     case is void.
    \item  Let $N$ be the redex $\letp{x}{\msec R}{Q}$ and
      let the contractum $M$ appears in $R$ as in the following
      reduction sequence
      \begin{equation*}
        \begin{split}
      &P = F[\letp{x}{\msec R'}{Q}]\\
      \reducif{i}\;&  P_1 =  F[\letp{x}{\msec R}{Q}]\\
      \reducif{j}\;& P_2 = F[\sub{Q}{R/x}]        
        \end{split}
      \end{equation*}
      By well-formedness, $x$ occurs exactly once in $Q$. This implies
       that applying first $P \reduci{j} P'$ cannot discard
      the redex in $R'$. Hence, we can produce the following shallow-first sequence of the same
      length:
      \begin{equation*}
        \begin{split}
      P = F[\letp{x}{\msec R'}{Q}] \reducif{j}\;&
      P'=F[\sub{Q}{R'/x}]\\
      \reducif{i}\;& P_2 = F[\sub{Q}{R/x}]
        \end{split}
      \end{equation*}
      Moreover, the reduction $P' \reducif{i} P_2$ must be outer-bang
      for $x$ cannot occur in a $\oc$-term in $Q$.\qedhere
  \end{enumerate}
\end{proof}
There are two notable differences with
Terui's swapping procedure. First, our procedure returns sequences of exactly the
same length as the original ones while his may return longer sequences. The reason is
that outer-bang contexts force
redexes to be duplicated before being reduced, as in
reduction~\eqref{dupget}, hence our swapping procedure cannot lengthen
sequences more. The other difference is that his calculus is affine
whereas ours is strictly linear. Therefore his procedure might shorten
sequences by discarding redexes and this breaks the argument for \emph{strong} polynomial
termination. His solution is to
introduce an auxiliary calculus with explicit discarding for which
swapping lengthens sequences. This is at
the price of introducing commutation rules which require quite a lot
of extra work to obtain the simulation result. We conclude that
strict linearity brings major proof simplifications while we have seen
it does not cause a loss of expressivity if we use garbage collecting regions.

Using the swapping lemma, we show that any reduction sequence that uses
outer-bang evaluation contexts can be simulated by a shallow-first sequence.
  \begin{proposition}[Simulation by shallow-first]
    \label{prop-std}
    To any reduction sequence $P_1 \treducf P_n$ corresponds
     a shallow-first reduction sequence $P_1
    \treducf P_n$ of the same length.
  \end{proposition}
  \begin{proof}
    By simple application of the bubble sort algorithm: traverse the
    original sequence from $P_1$ to $P_n$, compare the depth of each
    consecutive reduction steps, swap them by Lemma~\ref{swap} if they
    are in deep-first order. Repeat the traversal until no swap is
    needed. Note that we never need to swap two reduction steps of the
    same depth, which implies that we never need to reverse the order
    of dependent side effects. For example, in Figure~\ref{fig-std},
    the sequence $P \reducif{2} P' \reducif{1} P'' \reducif{0}
    P'''$ is reordered into $P \reducif{0} C \reducif{1}  B
    \reducif{2} P'''$ by 3 traversals.
  \end{proof}
  \begin{figure}[ht]
$$
    \begin{array}{rcccccl}
      P &\reducif{2}& P' &\reducif{1}& P'' &\reducif{0}& P'''\\
      P &\reducif{1}& A &\reducif{2}&  P'' &\reducif{0}& P'''\\
      P &\reducif{1}& A &\reducif{0}&  B &\reducif{2}& P'''\\
      P &\reducif{0}& C &\reducif{1}&  B &\reducif{2}& P'''
    \end{array}
$$
    \caption{Reordering of $P \treducf
      P'''$ in shallow-first}
    \label{fig-std}
  \end{figure}

As an application, we show that the simulation result applies to a
call-by-value operational semantics that we define in Figure~\ref{cbv}. 
\begin{figure}[ht]
  $$
  \begin{array}{lrcl@{}}
    \textrm{-values}&V &::=& x \mid \star \mid r \mid \lambda x.M \mid \dagger V\\
    \textrm{-terms}&M &::=& V \mid  MM \mid \msec M \mid \letd{x}{M}{M}\\ 
    &&&\get{r} \mid \st{r}{M} \mid (M \para M)\\ 
    \textrm{-stores}&S &::=& \store{r}{V} \mid (S \para S)\\
    \textrm{-programs}&P &::=& M \mid S \mid (P \para P)\\
    \textrm{-contexts}&F_\sv &::=& [\cdot] \mid F_{\sv}M \mid VF_\sv
    \mid \msec F_\sv\\
    &&&\letd{x}{F_\sv}{M} \mid \st{r}{F_\sv}\\
    &&& (F_\sv \para P) \mid (P \para F_\sv)
  \end{array}
  $$
  $$
  \begin{array}{lrcl}
    \multicolumn{4}{c}{\textrm{-reduction rules-}}\\
    (\beta_\sv) &F_\sv[(\lambda x.M)V] &\reduc_{\sv}& F_\sv[\sub{M}{V/x}]\\
    (\oc_{\sv}) &F_\sv[\letm{x}{\oc V}{M}] &\reduc_{\sv}& F_\sv[\sub{M}{V/x}]\\
    (\msec_{\sv}) &F_\sv[\letp{x}{\msec V}{M}] &\reduc_{\sv}& F_\sv[\sub{M}{V/x}]\\
     (\getsfv) &F_\sv[\get{r}] \para \store{r}{V} &\reduc_{\sv}& F_\sv[V]\\
     (\setsfv) & F_\sv[\st{r}{V}] &\reduc_{\sv}& F_\sv[\star] \para
     \store{r}{V}\\
     (\gcsfv) & F_\sv[\star \para  M] &\reduc_{\sv}& F_\sv[M]
  \end{array}
  $$
  \caption{CBV syntax and operational semantics}
  \label{cbv}
\end{figure}
We revisit the syntax of programs with a notion of value $V$ that may be
a variable, unit, a region, a $\lambda$-abstraction or
a $\dagger$-value.  Terms and programs are defined as previously
(see Figure~\ref{language-syntax}) except that $\oc M$ cannot be
constructed unless $M$ is a value.
 Store assignments are restricted to values. Evaluation contexts
$F_\sv$ are left-to-right call-by-value (obviously we do not evaluate in
stores). The call-by-value reduction relation is denoted by $\reduc_{\sv}$
and is defined modulo $F_\sv$ and $\equiv$.

From a
programming viewpoint, we shall only duplicate values. This explains
why we do not want to construct $\oc M$ if $M$ is not a value.

Call-by-value contexts $F_\sv$ are outer-bang contexts since $F_\sv$ cannot be
decomposed as $E[\oc E']$. This allows the relation $\reducf$ 
 to contain the relation $\reduc_{\sv}$. As a result, we obtain the
 following corollary.
\begin{corollary}[Simulation of CBV]
  \label{cor-cbv}
    To any reduction sequence $P_1 \treduc_\sv P_n$ corresponds
     a shallow-first reduction sequence $P_1
    \treducf P_n$ of the same length.
\end{corollary}
Remark that we may obtain a non call-by-value sequence but that the
semantics of the program is preserved (we compute $P_n$).
\section{Polynomial soundness of shallow-first}
\label{sec-soundness}
In this section we prove that well-formed programs admit polynomial bounds with
 a shallow-first strategy. We stress that this subsection is similar
 to Terui's~\cite{Terui07}; the main
 difficulty has been to design the polynomial depth system such that
 we could adopt a similar proof method.

As a first step, we define an \emph{unfolding} transformation on
programs.
\begin{definition}\textit{(Unfolding)}
  The \emph{unfolding} at depth $i$ of a program
  $P$, written
  $\unfold{i}{P}$, is defined as follows:
  \begin{equation*}
\begin{array}{r@{\;\;}c@{\;\;}l}
  \unfold{i}{x} &=& x\\
  \unfold{i}{r} &=& r\\
  \unfold{i}{\star} &=& \star\\
  \unfold{i}{\lambda x.M} &=& \lambda x.\unfold{i}{M}\\
  \unfold{i}{MN} &=& \unfold{i}{M}\unfold{i}{N}\\\\
  \unfold{i}{\dagger M} &=& \left\lbrace
    \begin{array}{ll}
\dagger \unfold{i-1}{M} & \text{if } i>0\\      
\dagger M & \text{if } i=0
    \end{array}
 \right.\\\\
  \unfold{i}{\letd{x}{M}{N}} &=& \left\lbrace
    \begin{array}{l@{}l@{\!}}
 \text{if } i=0, M=\oc M' \text{ and } \dagger = \oc:\\
\letb{x}{\underset{k \text{ times}}{\underbrace{MM\ldots
      M}}}{\unfold{0}{N}}\\
\text{where } k =\fo{x}{\unfold{0}{N}}\\\\
\text{otherwise:}\\
    \letd{x}{\unfold{i}{M}}{\unfold{i}{N}}
    \end{array}\right.\\\\
      \unfold{i}{\get{r}} &=& \get{r}\\
      \unfold{i}{\st{r}{M}} &=& \st{r}{\unfold{i}{M}}\\
      \unfold{i}{\store{r}{M}} &=& \store{r}{\unfold{i}{M}}\\
      \unfold{i}{P_1 \para P_2} &=& \unfold{i}{P_1} \para
      \unfold{i}{P_2}
\end{array}
  \end{equation*}
\end{definition}
This unfolding procedure is intended to duplicate
statically the occurrences that will be duplicated by redexes
occurring at depth $i$. For example, in the following reductions
occurring at
depth 0:
\begin{equation*}
  \begin{split}
&P = \letb{x}{\oc M}{(\letb{y}{\oc x}{\msec(yy)} \para \letb{y}{\oc x}{\msec(yy)}})\\
&\treduci{0}\; \msec(MM) \para \msec(MM)  
 \end{split}
\end{equation*}
the well-formed program $P$ duplicates the occurrence $M$ four
times. We observe that the unfolding at depth $0$ of $P$
reflects this duplication:
\begin{equation*}
  \begin{split}
    \unfold{0}{P}
= \;&\letb{x}{\oc M\oc M\oc M\oc M}{}\\
&(\letb{y}{\oc x\oc x}{\msec(yy)}\para \letb{y}{\oc x\oc x}{\msec(yy)})
  \end{split}
\end{equation*}
Unfolded programs are not intended to be reduced. However,
the size of an unfolded program can be used as a non increasing
measure in the following way.
\begin{lemma}
\label{unfold-dec}
  Let $P$ be a well-formed program such that\\ $P \reduci{i} P'$. Then
  $\size{\unfold{i}{P'}} \leq \size{\unfold{i}{P}}$.
\end{lemma}
\begin{proof}
First, we assume the occurrences labelled with `$\para$' and `$r
\Leftarrow$' do not count in the size of a program and that
`$\mathsf{set}(r)$' counts for two occurrences, such that the size
strictly decreases by the rule $(\setsf)$. Then,
  it is clear that $(\oc)$ is the only reduction rule that can make
  the size of a program increase, so let 
  $$P = F[\letb{x}{\oc N}{M}] \reduci{i} P' = F[\sub{M}{N/x}]$$
  We have
  \begin{align*}
    \unfold{i}{P}
    =\;& F'[\letb{x}{\underset{n \textrm{ times}}{\underbrace{\oc N\oc N\ldots\oc
          N}}}{\unfold{0}{M}}]\\
  \unfold{i}{P'} =\;& F'[\unfold{0}{\sub{M}{N/x}}]
  \end{align*}
  for some context $F'$ and  $n = \fo{x}{\unfold{0}{M}}$. Therefore we are left to
  show $$\size{\unfold{0}{\sub{M}{N/x}}} \leq
  \size{\letb{x}{\underset{n \textrm{ times}}{\underbrace{\oc N\oc
          N\ldots\oc N}}}{\unfold{0}{M}}}$$
which is clear since $N$ must occur $n$ times in $\unfold{0}{\sub{M}{N/x}}$.
\end{proof}
We observe in the following lemma that the size of an unfolded program bounds
quadratically the size of the original program.
\begin{lemma}
\label{presq}
  If $P$ is well-formed, then for any depth $i \leq d(P)$:
  \begin{enumerate}
  \item\label{presq1} $\foa{\unfold{i}{P}} \leq \size{P}$,
  \item\label{presq2} $\size{\unfold{i}{P}} \leq \size{P} \cdot (\size{P}
    -1)$,
  \end{enumerate}
\end{lemma}
\begin{proof}By induction on $P$ and $i$.
\end{proof}

We can then bound the size of a program after reduction.
\begin{lemma}[Squaring]
\label{squaring}
  Let $P$ be a well-formed program such that $P \treduci{i} P'$. Then:
  \begin{enumerate}
  \item $\size{P'} \leq \size{P} \cdot (\size{P} -1)$
  \item the length of the sequence is bounded by $\size{P}$
  \end{enumerate}
\end{lemma}
\begin{proof}\hfill
  \begin{enumerate}
  \item
    \label{squaring1} 
    By Lemma~\ref{unfold-dec} it is clear that $\size{\unfold{i}{P'}}
    \leq \size{\unfold{i}{P}}$. Then by Lemma~\ref{presq}-\ref{presq2}
    we obtain $\size{\unfold{i}{P'}}\leq \size{P} \cdot (\size{P}
    -1)$. Finally it is clear that $\size{P'} \leq
    \size{\unfold{i}{P'}}$ thus $\size{P'} \leq \size{P} \cdot
    (\size{P} -1)$.
  \item\label{linearseq} It suffices to remark  $\size{P'}_i <
    \size{P}_i \leq \size{P}$.\qedhere
  \end{enumerate}
\end{proof}
Finally we obtain the following theorem for a shallow-first strategy
using any evaluation context.
\begin{theorem}[Polynomial bounds]
\label{thm-bound}
Let $P$ be a well-formed program such that $d(P)=d$ and  $P \treduc
P'$ is shallow-first. Then:
\begin{enumerate}
\item $\size{P'} \leq \size{P}^{2^d}$
\item the length of the reduction sequence is bounded by $\size{P}^{2^d}$
\end{enumerate}
\end{theorem}
\begin{proof}
  The reduction $P \treduc P'$ can be decomposed as
$
P = P_0 \treduci{0} P_1 \treduci{1} \ldots \treduci{d-1} P_d \treduci{d}
P_{d+1}=P'
$.
To prove (1), we observe that by iterating Lemma~\ref{squaring}-\ref{squaring1} we obtain $\size{P_d} \leq
\size{P_0}^{2^d}$. Moreover it is clear that $\size{P_{d+1}} \leq \size{P_{d}}$. Hence  $\size{P'}
\leq \size{P}^{2^d}$. To prove (2), we first prove by induction on $d$
that $\size{P_0} + \size{P_1} + \ldots + \size{P_d} \leq
\size{P_0}^{2^d}$. 
By Lemma~\ref{squaring}-\ref{linearseq}, it is clear that the length
of the reduction  $P \treduc P'$ is bounded by   $\size{P_0} +
\size{P_1} + \ldots + \size{P_d}$, which is in turn bounded by $\size{P_0}^{2^d}$.
\end{proof}
It is worth noticing that the first bound takes the size of
all the threads into account and that the second bound is valid for any thread
interleaving.
\begin{corollary}[Call-by-value is polynomial]
  \label{zz}
     The call-by-value evaluation of a well-formed program
    $P$ of size $n$ and depth $d$ can be
    computed in time $O(n^{2^d})$.
\end{corollary}
\begin{proof}
  Let $P \treduc_{\sv} P'$ be the call-by-value  reduction sequence of the
  well-formed program $P$. By Corollary~\ref{cor-cbv} we can reorder
  the sequence into a shallow-first  sequence $P \treducf
  P' $ of the same length. By Theorem~\ref{thm-bound} we know that its length
  is bounded by $\size{P}^{2^d}$ and that $\size{P'} \leq \size{P}^{2^d}$.
\end{proof}


\section{A  polynomial type system}
\label{sec-types}
The depth system entails termination in polynomial time but does {\em
  not} guarantee that  programs `do not go wrong'.  In particular, the well-formed
program in~\eqref{blocked}
get  stuck on a non-value.  In this section, we propose a solution to this problem by
introducing a polynomial type system as a simple decoration of the
polynomial depth system with linear types. Then, we derive a progress
proposition which guarantees that well-typed programs cannot deadlock (except
when trying to read an empty region).

We define the syntax of types and contexts in Figure~\ref{syn-type}.
\begin{figure}[ht]
$$
\begin{array}{lrcl@{}}
 \textrm{-type variables}&\multicolumn{3}{l}{t,t',\ldots} \\
       \textrm{-types}&\alpha&::=& \behtype \mid A                       \\
\textrm{-res. types}&A&::=& t \mid \tertype \mid A\limpe{e_1}{e_2} \alpha \mid \dagger A
\mid \forall t.A
\mid \rgtype{r}{A}     \\
\textrm{-var. contexts}&\Gamma &::=&
\hyp{x_{1}}{u_{1}}{A_{1}}{e_1},\ldots,\hyp{x_{n}}{u_{n}}{A_{n}}{e_n}
 \\
\textrm{-reg. contexts}&R&::=& r_1:(\delta_1,A_1),\ldots,r_n:(\delta_n,A_n)  
\end{array}
$$
\caption{Syntax of types, effects and contexts}
\label{syn-type}
\end{figure}
Types are denoted with $\alpha,\alpha',\ldots$. Note that we
distinguish a special \emph{behaviour} type $\behtype$ which is given to the
entities of the language which are not supposed to return a result
(such as a store or several terms in parallel) while types of entities
that may return a result are denoted with $A$.  Among the types $A$, we
distinguish type variables $t, t', \ldots$, a terminal type
$\tertype$, a linear functional type $A \limp \alpha$, the type
$\bang A$ of terms of type $A$ that may be duplicated, the type
$\msec A$ of terms of type $A$ that may have been duplicated, the type
$\forall t.A$ of polymorphic terms and the type $\rgtype{r}{A}$ of 
regions $r$ containing terms of type $A$. Hereby types may depend on
regions.

In contexts, usages play the same role as in the
depth system. Writing $x : (u, A)$ means that the variable $x$
ranges on terms of type $A$ and can be bound according to $u$. Writing
$r:(\delta,A)$ means that the region $r$ contain
terms of type $A$ and that $\get{r}$ and $\st{r}{M}$ may only occur at
depth $\delta$. The typing system will additionally guarantee that
whenever we use a type $\rgtype{r}{A}$ the region context contains a
hypothesis $r:(\delta,A)$.

Because types depend on regions, we have to be careful in stating in
Figure~\ref{types-contexts} when a region-context and a type are
compatible ($R\downarrow \alpha$), when a region context is
well-formed ($R\vdash$), when a type is well-formed in a region
context (\mbox{$R\vdash \alpha$}) and when a context is well-formed in
a region context (\mbox{$R\vdash \Gamma$}).  A more informal way to
express the condition is to say that a judgement
$r_{1}:(\delta_1,A_{1}),\ldots,r_{n}:(\delta_n,A_{n}) \vdash \alpha$ is well formed provided
that:
$(1)$ all the region constants
occurring in the types $A_1,\ldots,A_n,\alpha$ belong to the set
$\set{r_1,\ldots,r_n}$, 
$(2)$ all types of the shape
$\rgtype{r_{i}}{B}$ with $i\in \set{1,\ldots,n}$ and occurring in the
types $A_1,\ldots,A_n,\alpha$ are such that \mbox{$B=A_i$}.
\begin{figure}[ht]
$$
\begin{array}{c@{}}
\inference
{}
{R \downarrow t}

\quad

\inference
{}
{R \downarrow \tertype}
\quad

\inference
{}
{R \downarrow \behtype}
\quad

\inference
{R\downarrow A & R\downarrow \alpha}
{R\downarrow (A\limpe{e_1}{e_2} \alpha)}

\\\\
\inference
{R\downarrow A}
{R\downarrow \dagger A}

\quad

\inference
{r:(\delta,A) \in R}
{R\downarrow \rgtype{r}{A}}

\quad

\inference
{R \downarrow A & t \notin R}
{R \downarrow \forall t.A}

\\\\
\inference
{\forall r:(\delta,A) \in R \\ R\downarrow A}
{R\vdash}

\quad

\inference
{R\vdash & R\downarrow \alpha}
{R\vdash \alpha} 

\quad
\inference
{\forall \hyp{x}{\delta}{A}{e}\in \Gamma \\ R\vdash A }
{R\vdash \Gamma}

\end{array}
$$
\caption{Types and contexts}
\label{types-contexts}
\end{figure}
\begin{example}
One may verify that the judgment $r:(\delta,\tertype \limp \tertype)
\vdash \rgtype{r}{(\tertype \limp \tertype)}$ can be derived while
 judgements
$r:(\delta,\tertype) \vdash \rgtype{r}{(\tertype \limp \tertype)}$
and
$r:(\delta,\rgtype{r}{\tertype}) \vdash \tertype$ cannot.
\end{example}
We notice the following substitution property on types.
\begin{proposition}
\label{type-substitution}
If $R\vdash \forall t.A$ and $R\vdash B$ then 
$R\vdash \sub{A}{B/t}$.
\end{proposition}
A typing judgement takes the form: $ R;\Gamma \vdash^{\delta} P:\alpha
$.  It attributes a type $\alpha$ to the program $P$ occurring at depth
$\delta$, according to region context $R$ and variable context $\Gamma$.
Figure~\ref{poly-type-system} introduces the polynomial type system.
\begin{figure*}[ht]
$$
\begin{array}{c}
  \inference
  {R\vdash\Gamma & x:(\uaff,A) \in \Gamma }
  {R;\Gamma  \vdash^\delta x:A}
\;
\inference
{R\vdash\Gamma }
{R;\Gamma \vdash^\delta \star:\tertype} 

\;

\inference
{R\vdash \Gamma}
{R;\Gamma  \vdash^\delta r:\rgtype{r}{A}}\quad 

  \inference
  {  \fo{x}{M} = 1 
    \\R;\Gamma,x:(\uaff,A)  \vdash^\delta M:\alpha}
  {R:\Gamma \vdash^\delta \lambda x.M: A \limpe{e_1}{e_2} \alpha}

  \\\\

  \inference
  {R;\Gamma \vdash^\delta M:A\limpe{e_1}{e_2} \alpha \\ R;\Gamma \vdash^\delta
    N:A}
  {R;\Gamma \vdash^\delta MN:\alpha}\qquad

  \inference
  { \foa{M} \leq 1 \\R;\Gamma_{\uaff} \vdash^{\delta+1} M:A }
  {R; \Gamma_{\ubang},\Delta_{\upara}, \Psi_{\uaff} \vdash^\delta \oc M:\oc A}

  \;

  \inference
  {R;\Gamma \vdash^\delta M:\oc A & \fo{x}{N} \geq 1 \\ R;\Gamma,x:(\ubang,A) \vdash^\delta
    N:\alpha}
  {R;\Gamma \vdash^\delta \letb{x}{M}{N}:\alpha}\quad

  \inference
  { R;\Gamma_{\uaff},\Delta_{\uaff} \vdash^{\delta+1} M:A}
  { R;\Gamma_{\upara},\Delta_{\ubang},\Psi_{\uaff} \vdash^\delta \msec M: \msec A}

  \\\\

  \inference
  {R;\Gamma \vdash^\delta M:\msec A & \fo{x}{N} = 1\\ R;\Gamma,x:(\upara,A) \vdash^\delta
    N:\alpha }
  {R;\Gamma \vdash^\delta \letp{x}{M}{N}:\alpha}\qquad

\inference
{ t \notin (R;\Gamma)\\R;\Gamma \vdash^\delta M : A}
{R;\Gamma \vdash^\delta M: \forall t.A}

\;

\inference
{R;\Gamma \vdash^\delta M : \forall t.A & R \vdash B}
{R;\Gamma \vdash^\delta M : \sub{A}{B/t}}

\quad

\inference
{R\vdash\Gamma &r:(\delta,A)\in R}
{R;\Gamma  \vdash^\delta \get{r}:A}
\\\\

\inference
{r:(\delta,A)\\ R;\Gamma  \vdash^\delta
M : A}
{R;\Gamma  \vdash^\delta \st{r}{M}:\tertype} 

\;

\inference
{r:(\delta,A) \\ R;\Gamma  \vdash^\delta M:A}
{R;\Gamma \vdash^0  \store{r}{M} : \behtype}

\qquad

\inference
{R;\Gamma  \vdash^\delta P_1:\tertype \text{ or } P_1 = S\\
  R;\Gamma  \vdash^\delta P_2:\alpha}
{R;\Gamma  \vdash^\delta (P_1\para P_2):\alpha}


\;

\inference
{
  R;\Gamma  \vdash^\delta P_i:\alpha_i}
{R;\Gamma  \vdash^\delta (P_1\para P_2):\behtype}  
\end{array}
$$
\caption{A polynomial type system}
\label{poly-type-system}
\end{figure*}
We comment on some of the rules. A $\lambda$-abstraction may only take
a term of result-type as argument, \ie two threads in parallel are not
considered an argument. The typing of $\dagger$-terms is
limited to result-types for we may not duplicate several threads in parallel.
There exists  two rules for typing parallel programs. The one on
the left indicates that a program $P_2$ in parallel with a store or a
 thread producing a terminal value  should have the type of $P_2$ since we might
be interested in its result (note that we omit the symmetric rule for
the program $(P_2 \para P_1)$). The one on the right
indicates that two programs in parallel cannot reduce to a single
result.
\begin{example}
  The program of Figure~\ref{trees} is well-typed
  according to the following derivable judgement:
$$
R;- \vdash^\delta \letb{x}{\get{r}}{\st{r}{(\oc
        x)(\msec x)}} \para \store{r}{\oc (\lambda x.x\star)} : \tertype
$$
where $R=r:(\delta,\forall t.\oc((\tertype \multimap t) \multimap
t))$. Whereas the program in~\eqref{blocked} is not.
\end{example}


\begin{remark}
  We can easily see that a well-typed program is also well-formed.
\end{remark}

The polynomial type system enjoys the subject reduction property for
the largest relation $\reduc \supseteq \reducf \supseteq \reduc_{\sv}$.
\begin{lemma}[Substitution]\hfill
  \label{substitution-types}
  \begin{enumerate}
  \item\label{sub1t}   If $R;\Gamma,\hyp{x}{\uaff}{A}{e_1} \vdash^{\delta} M : B$ and
  $R;\Gamma \vdash^\delta N:A$ then $R;\Gamma \vdash^\delta
  \sub{M}{N/x}:B$.
  \item\label{sub2t}   If $R;\Gamma,\hyp{x}{\upara}{A}{e_1} \vdash^{\delta} M : B$ and
  $R;\Gamma \vdash^\delta \msec N:\msec A$ then $R;\Gamma \vdash^\delta
  \sub{M}{N/x}:B$.
  \item\label{sub3t}   If $R;\Gamma,\hyp{x}{\ubang}{A}{\emptyset} \vdash^{\delta} M : B$ and
  $R;\Gamma \vdash^\delta \oc N:\oc A $ then $R;\Gamma \vdash^\delta
  \sub{M}{N/x}:B$.
  \end{enumerate}
\end{lemma}
\begin{proposition}[Subject Reduction]
\label{subred-types} If $R;\Gamma
  \vdash^\delta P:\alpha$ and $P\reduc P'$ then $R;\Gamma
  \vdash^\delta P':\alpha$.
\end{proposition}

\vspace{3cm}
Finally, we establish a progress proposition which shows that any
well-typed call-by-value program (\ie defined from Figure~\ref{cbv}) reduces to several
threads in parallel which are values or deadlocking reads. 
\begin{proposition}[Progress]
\label{progress}
Suppose $P$ is a closed typable call-by-value program which cannot
reduce. Then $P$ is structurally equivalent to a program
\[
M_1 \para \cdots \para M_m \para S_1 \para \cdots
\para S_n\quad m,n\geq 0
\]
where $M_i$ is either a value or can only be decomposed as a term
$F_\sv[\get{r}]$ such that no value is
associated with the region $r$ in the stores $S_1,\ldots,S_n$.
\end{proposition}

\section{Expressivity}
\label{sec-expr}
We now illustrate the expressivity of the polynomial type
system. First we show that our system is complete in the
extensional sense: every polynomial time function can be represented (Subsection~\ref{subsec-complete}). Then
we introduce a language with memory locations representing higher-order references for which
the type system can be easily adapted (Subsection~\ref{subsec-references}). Building
on this language, we give an example of polynomial programming (Subsection~\ref{subsec-programming}).

As a first step, we define some
Church-like encodings in Figure~\ref{church} where we abbreviate $\lambda x.\letd{x}{x}{M}$ by
$\lambda^\dagger x.M$. 
\begin{figure}[ht]
  $$
  \begin{array}{@{}rcl@{}}
    \nat &=& \forall t.\oc(t \multimap t) \multimap \msec (t \multimap t)\\    
\churchn{n} &:& \nat\\  
    \churchn{n} &=& \lambda^{\oc} f.\msec(\lambda x.\underset{n \text{
        times}}{\underbrace{f(\ldots(f}}x)))\\\\



    \add &:& \nat \multimap \nat \multimap \nat\\
    \add &=& \lambda m.\lambda n.\lambda^{\oc} f.\letp{y}{m\oc
        f}{}\\
    &&\letp{z}{n\oc f}{\msec(\lambda x.y(zx))}\\\\



    \bnat &=& \forall t.\oc(t \multimap t) \multimap \oc(t \multimap
    t) \multimap  \msec (t \multimap t)\\    
 \multicolumn{3}{l}{\text{for } w=i_0\ldots i_n \in \set{0,1}^{*}}  \\
 \churchn{w} &:& \bnat\\
    \churchn{w} &=& \lambda^{\oc} x_0.\lambda
    x_1^{\oc}.\msec(\lambda z.x_{i_0}(\ldots (x_{i_n}z)))\\\\

    \listtype{A} &=& \forall t.\oc(A \multimap t \multimap t)
    \multimap \msec (t \multimap t)\\
    \churchl{u_1,\ldots,u_n} &:& \listtype{A}\\
    \churchl{u_1,\ldots,u_n} &=& \lambda f^{\oc}.\msec (\lambda
      x.fu_1(fu_2\ldots(fu_nx)))\\\\

    \lit &:& \forall u. \forall t.\oc(u\multimap t\multimap t)
    \multimap \listtype{u} 
      \multimap \msec t \multimap \msec t\\
    \lit &=&  \lambda f.\lambda l.\lambda^{\msec} x.\letp{y}{lf}{\msec(yx)}



  \end{array}
  $$
  \caption{Church encodings}
  \label{church}
\end{figure}
We have natural numbers of type $\nat$, binary natural number of type
$\bnat$ and lists of type $\listtype{A}$ that contain values of type $A$.
\subsection{Polynomial completeness}
\label{subsec-complete}
The representation of polynomial functions relies on the representation of binary
words.  The precise notion of representation is spelled out in the
following definitions.

\begin{definition}\textit{(Binary word representation)}
  Let $- \vdash^\delta M : \msec^p\bnat$ for some $\delta,p
  \in \mathbb{N}$. We say $M$ \emph{represents} $w \in \set{0,1}^{*}$,
  written $M \Vdash w$, if $M \treduc \msec^p\churchn{w}$.
\end{definition}
\begin{definition}\textit{(Function representation)}
  \label{def-function-rep}
  Let $- \vdash^\delta F: \bnat \multimap \msec^d\bnat$ where
  $\delta,d \in \mathbb{N}$ and $f:\set{0,1}^{*} \to \set{0,1}^{*}$. We say
$F$ \emph{represents} $f$, written $F \Vdash f$,
if for any $M$ and $w \in \set{0,1}^{*}$  such that $-
\vdash^\delta M:\bnat$ and $M \Vdash w$,
$FM \Vdash f(w)$. 
\end{definition}
The following theorem is a restatement of Girard~\cite{LLL} and Asperti~\cite{Asperti98}.
\begin{theorem}[Polynomial completeness]\hfill\\
    Every function $f:\set{0,1}^{*} \to \set{0,1}^{*}$ which can be
    computed by a Turing machine in time bounded by a polynomial of
    degree $d$ can be represented by a term of type $\bnat
    \multimap \msec^d\bnat$.
\end{theorem}

\subsection{A language with higher-order references}
\label{subsec-references}
Next, we give an application of the language
 with abstract regions by presenting a connection with a
language with dynamic memory locations representing higher-order references. 

The differences with the region-based system are presented in
Figure~\ref{concrete}.
\begin{figure}[ht]
  $$
  \begin{array}{rcl}
    M &::=& \ldots \mid \nu x.M\\
    F_{\nu} &::=& F_{\sv} \mid \nu x.F_{\nu}
  \end{array}
  $$
  $$
  \begin{array}{lr@{\!\!\!\!\!\!}c@{\!\!\!\!\!\!}l}
    (\nu) & F_\nu[\nu x.M] &\equiv& \nu x.F_\nu[M]\\
    & &\text{if }x \notin \fv{F_\nu}&\\\\
     (\getsf_{\nu}) &F_\nu[\get{x}] \para \store{x}{V} &\reduc_{\nu}& F_\nu[V] \para \store{x}{V}\\
     (\setsf_{\nu}) & F_\nu[\st{x}{V}] \para \store{x}{V'} &\reduc_{\nu}& F_\nu[\star] \para
     \store{x}{V}
  \end{array}
  $$
$$
\begin{array}{c}

\inference
{R;\Gamma, \hyp{x}{u}{\regtype{r}{\oc A}}{} \vdash^\delta M : B}
{R;\Gamma \vdash^\delta \nu x.M : B}
\;
\inference
{R(r) = (\delta,\oc A)\\R;\Gamma\vdash^\delta x : \regtype{r}{\oc A}}
{R;\Gamma  \vdash^\delta \get{x}:\oc A}
\\\\

\inference
{R(r)=(\delta,\oc A)\\R;\Gamma\vdash^\delta x : \regtype{r}{\oc A}\\ R;\Gamma  \vdash^\delta
M : \oc A}
{R;\Gamma  \vdash^\delta \st{x}{M}:\tertype} 

\quad

\inference
{R(r)=(\delta,\oc A)\\R;\Gamma\vdash^\delta x : \regtype{r}{\oc A} \\ R;\Gamma  \vdash^\delta
  V:\oc A}
{R;\Gamma \vdash^0  \store{x}{V} : \behtype}  
\end{array}
$$
  \caption{A call-by-value system with references}
  \label{concrete}
\end{figure}
We introduce terms of the form $\nu x.M$ to
generate a fresh memory location $x$ whose scope is $M$. Contexts are
call-by-value and allow evaluation under $\nu$ binders. The structural rule $(\nu)$ is for
scope extrusion. Region constants have been removed from the syntax
of terms hence reduction rules $(\getsf_{\nu})$ and $(\setsf_{\nu})$
relate to memory locations. The operational semantics of references
is adopted: when assigning a value to
a memory location, the previous value is \emph{overwritten}, and when
reading a memory location, the value is \emph{copied} from the store. 
We see in the typing rules that region constants still appear in region
types and that a memory location must be a free variable that relates to an
abstract region $r$ by having the type $\regtype{r}{A}$.

There is a simple translation from the language with memory locations
to the language with regions. It consists in replacing
 the (free or bound) variables with a region type of the shape
$\regtype{r}{A}$ by the constant $r$. We then observe that 
 read access and assignments to references are
 mapped to several reduction steps in the system with regions. It
 requires the following observation: in the
typing rules, memory locations only relate to regions with duplicable content of
type $\oc A$. This allows us to simulate the
\emph{copy from memory} mechanism of references by decomposing it into a \emph{consume} and
\emph{duplicate} mechanism in the language with regions. More precisely:
 an occurrence of $\get{x}$
where $x$ relates to region $r$ is translated into
\begin{equation*}
  \begin{split}
\letb{y}{\get{r}}{\st{r}{\oc y} \para \oc y}
  \end{split}
\end{equation*}
such that     
\begin{equation*}
  \begin{split}
&F_\sv[\letb{y}{\get{r}}{\st{r}{\oc y} \para \oc y}] \para \store{r}{\oc V}\\
\preduc_{\sv}\;& F[\oc V] \para \store{r}{\oc V}  
  \end{split}
\end{equation*}
simulates the reduction $(\getsf_{\nu})$.
Also, it is easy to see that a reduction step $(\setsf_{\nu})$ can be simulated
by exactly one reduction step $(\setsfv)$. Since typing is preserved
by translation, we conclude that any time complexity bound can be lifted to the language with
 references. 

Note that this also works if we adopt
the operational semantics of communication channels; in that case, memory
locations can also relate to regions containing non-duplicable content
since reading a channel means \emph{consuming} the value.

\subsection{Polynomial programming}
\label{subsec-programming}
Using higher-order references,
we show that it is possible to program the iteration of operations producing a side effect on an
inductive data structure, possibly in parallel.

Here is the function $\mathsf{update}$ taking as argument a memory location $x$ related
  to region $r$ and incrementing the numeral stored at that location:
$$
\begin{array}{l@{}}
  \hyp{r}{3}{\oc\nat}{\set{r}}; - \vdash^2
  \mathsf{update} : \oc\regtype{r}{\oc\nat} \limpe{\emptyset}{\emptyset}
  \msec \tertype \multimap \msec \tertype\\
  \mathsf{update} = \lambda^\oc x.\lambda^\msec z.\msec
      (\st{x}{\letb{y}{\get{x}}{\oc(\add \; \churchn{2}\; y)}} \para z)
\end{array}
$$
The second argument $z$ is to be garbage collected.
Then we define the program $\mathsf{run}$ that iterates the function
$\mathsf{update}$ over a list $\churchl{\oc x, \oc y, \oc z}$ of 3 memory locations:
$$
\begin{array}{l}
   \hyp{r}{3}{\oc \nat}{\set{r}} \vdash^1 \mathsf{run} : \msec
   \msec \tertype\\
   \mathsf{run} = \lit \; \oc \mathsf{update} \;
   \churchl{\oc x,\oc y,\oc z} \; \msec\msec \star
\end{array}
  $$
All addresses have type $\oc \regtype{r}{\oc \nat}$ and thus relate to
  the same region $r$. Finally, the
  program $\mathsf{run}$ in parallel with some store assignments reduces as expected:
 $$
  \begin{array}{c}
  \mathsf{run}  \para  \store{x}{ \oc\churchn{m}} \para
   \store{y}{\oc\churchn{n}} \para \store{z}{ \oc\churchn{p}}\\
  \treduc_{\nu} \msec\msec \star \para  \store{x}{ \oc\churchn{2+m}} \para
  \store{y}{ \oc\churchn{2+n}} \para \store{z}{ \oc\churchn{2+p}}
  \end{array}
  $$
Note that due to the Church-style encoding of numbers and lists, we
assume that the relation $\reduc_{\nu}$ may reduce under binders when required.

Building on this  example, suppose we want to write a program
of three threads where each thread concurrently 
increments the numerals  pointed by the memory locations of the list.
Here is the function $\mathsf{gen\_threads}$ taking a functional $f$
and a value $x$ as arguments and generating three
threads where $x$ is applied to $f$:
$$
  \begin{array}{l}
    \hyp{r}{3}{\oc\nat}{\set{r}} \vdash^0 \mathsf{gen\_threads} : \forall t.\forall t'.\oc(t \limpe{\emptyset}{\set{r}} t')
    \limpe{\emptyset}{\emptyset} \oc t \limpe{\emptyset}{\set{r}} \behtype\\
    \mathsf{gen\_threads} = \lambda^\oc f.\lambda^\oc
      x.\msec(fx) \para \msec (fx) \para \msec (fx)
  \end{array}
  $$
  We define the functional $\mathsf{F}$ like $\mathsf{run}$ but parametric in the
  list:
 $$
  \begin{array}{l}
    \hyp{r}{3}{\oc\nat}{\set{r}} \vdash^1 \mathsf{F} : \listtype{\oc\regtype{r}{\oc\nat}} \limpe{\emptyset}{\set{r}} \msec\msec \tertype\\
    \mathsf{F} = \lambda l.\lit \; \oc \mathsf{update} \; l \; \msec\msec \star
  \end{array}
  $$
  Finally the concurrent iteration is defined in $\mathsf{run\_threads}$:
  $$
  \begin{array}{l}
      \hyp{r}{3}{\oc\nat}{\set{r}} \vdash^0 \mathsf{run\_threads} : \behtype\\
    \mathsf{run\_threads} = \mathsf{gen\_threads} \; \oc \mathsf{F} \;
    \oc\churchl{\oc x,\oc y,\oc z}
  \end{array}
  $$
  The program is well-typed for side effects occurring at depth
  $3$ and it reduces as follows:
  $$
  \begin{array}{l}
  \mathsf{run\_threads} \para  \store{x}{ \oc\churchn{m}} \para
  \store{y}{ \oc\churchn{n}} \para \store{z}{ \oc\churchn{p}}\\
  \treduc_{\nu}  \msec\msec\msec \star\para \store{x}{ \oc\churchn{6+m}} \para
  \store{y}{ \oc\churchn{6+n}} \para \store{z}{ \oc\churchn{6+p}}
  \end{array}
  $$
Note that different thread interleavings are possible but in this
particular case they are confluent.

\section{Conclusion and Related work}
We have proposed a type system for a higher-order functional language
with multithreading and side effects that guarantees termination in
polynomial time, covering any scheduling of threads and  taking
account of thread generation. To the best of our knowledge, there appears to be no other
characterization of polynomial time in such a language.
The polynomial soundness of the call-by-value strategy relies on 
the simulation of call-by-value by a shallow-first strategy which is
proved to be polynomial. The proof is a significant adaptation of Terui's
methodology~\cite{Terui07}: it is greatly simplified by a strict linearity condition
and based on a clever analysis of the evaluation order of side effects
which is shown to be preserved.

\paragraph{Related work}
The framework of light
logics has been previously applied to a higher-order
$\pi$-calculus~\cite{LagoExpress} and a functional language with
pattern-matching and recursive definitions~\cite{BaillotGM10}.
 The notion of
\emph{stratified region}\footnote{Here we speak of stratification by
   means of a type-and-effect discipline, this is not to be confused
   with the notion of stratification \emph{by depth level} that is used
   in the present paper.}
has been proposed~\cite{Boudol10,Amadio09} to
ensure the termination of a higher-order multithreaded language with
side effects . In the setting of
\emph{synchronous} computing, static analyses have been developed to
bound resource consumption in a synchronous
$\pi$-calculus~\cite{AmadioD07} and a multithreaded first-order
language~\cite{AmadioD06}. Recently, the framework of \emph{complexity
  information flow} have been applied to characterize
polynomial multithreaded imperative programs~\cite{MarionPechoux}.


\paragraph{Acknowledgments} The author wishes to thank Roberto Amadio
for his precious help on the elaboration of this work and Patrick
Baillot for his careful reading of the paper.
 
\bibliographystyle{abbrvnat} 
\renewcommand{\bibfont}{\normalsize}
\bibliography{references}
\end{document}